\documentclass{article}

    \usepackage[preprint]{neurips_2019}

\usepackage[utf8]{inputenc} 
\usepackage[T1]{fontenc}    
\usepackage{hyperref}       
\usepackage{url}            
\usepackage{booktabs}       
\usepackage{graphicx} 
\usepackage{amsfonts}       
\usepackage{nicefrac}       
\usepackage{microtype}      
\usepackage{natbib}
\usepackage{amsmath}
\usepackage{amstext}
\usepackage{algorithm}
\usepackage[noend]{algpseudocode}
\usepackage{amsthm}

\usepackage{xcolor}
\usepackage{tabularx}
\usepackage{booktabs}

\newtheorem{theorem}{Theorem}
\newtheorem{defn}{Definition}
\newtheorem{corr}{Corollary}

\newtheorem{lemma}{Lemma}
\newtheorem{example}{Example}

\newtheorem{knownresult}{Known Result}

\usepackage{mathtools}

\DeclareMathOperator*{\argmax}{arg\,max}

\newcommand{\defeq}{\mathrel{:\mkern-0.25mu=}}
\newcommand{\eqdef}{\mathrel{=\mkern-0.25mu:}}
\newcommand{\expectation}{\ensuremath{\mathbb{E}}}

\makeatletter
\g@addto@macro\normalsize{%
  \setlength\abovedisplayskip{4pt}
  \setlength\belowdisplayskip{4pt}
  \setlength\abovedisplayshortskip{4pt}
  \setlength\belowdisplayshortskip{4pt}
}
\makeatother

\title{Fair Division Without Disparate Impact}

\author{%
  Alexander Peysakhovich* \\
  Facebook AI Research
  \And
  Christian Kroer* \\
  Facebook Core Data Science
}

\begin{document}

\maketitle

\begin{abstract}
We consider the problem of dividing items between individuals in a way that is fair both in the sense of distributional fairness and in the sense of not having disparate impact across protected classes. An important existing mechanism for distributionally fair division is competitive equilibrium from equal incomes (CEEI). Unfortunately, CEEI will not, in general, respect disparate impact constraints. We consider two types of disparate impact measures: requiring that allocations be similar across protected classes and requiring that average utility levels be similar across protected classes. We modify the standard CEEI algorithm in two ways: equitable equilibrium from equal incomes, which removes disparate impact in allocations, and competitive equilibrium from equitable incomes which removes disparate impact in attained utility levels. We show analytically that removing disparate impact in outcomes breaks several of CEEI's desirable properties such as envy, regret, Pareto optimality, and incentive compatibility. By contrast, we can remove disparate impact in attained utility levels without affecting these properties. Finally, we experimentally evaluate the tradeoffs between efficiency, equity, and disparate impact in a recommender-system based market.
\end{abstract}

%
%



%
%


\maketitle

\section{Introduction}
Allocating a finite supply of items across individuals with heterogeneous preferences is an extremely important practical problem \citep{roth2002economist,roth2015gets}. In all applications mechanism designers need to trade off between various quantities that cannot be satisfied simultaneously and the most typical debate is how to trade off efficiency (total realized utility), distributional concerns, and various incentive properties \citep{nisan2001algorithmic,chen2013truth,caragiannis2016unreasonable}. A more recent literature in algorithmic fairness has taken up a different notion of fairness - ensuring `equal' treatment across groups of individuals \citep{barocas2016big,feldman2015certifying}. In this paper we explore how to incorporate these group-based notions into a popular division mechanism: competitive equilibrium from equal incomes (CEEI, \cite{varian1974equity,budish2011combinatorial}).

CEEI considers a set of items with finite supply and a set of individuals that have preferences over the items. In CEEI individuals are allocated a budget of pseudo-currency and individual valuations are used to construct demand functions i.e. what items individuals would `buy' using their budget given a price for each item. The mechanism computes a market equilibrium - a set of prices for each item such that the sum of demands meets supply. The allocations from this equilibrium is how the items are split.

Since CEEI is based on market equilibrium, the division inherits its properties: it is envy-free (everyone prefers their own allocation to anyone else's), has no regret (individuals would not change the allocations, given the budget and prices), it is Pareto efficient (nobody can be made worse off without making someone else better off), and it is strategy proof (individuals report their valuations truthfully) when instances are large \citep{varian1974equity,azevedo2018strategy}. 

An equally important question to what to maximize is how to actually compute CEEI \citep{nisan2001algorithmic}. When items are divisible we can use convex programming to choose a particular equilibrium solution where the allocation maximizes the product of utilities aka. the Nash social welfare \citep{eisenberg1959consensus,cole2017convex}. In practice NSW maximization produces highly efficient and yet distributionally fair outcomes and has been hailed as an `unreasonably fair'  way of allocating items, even in the indivisible case~\citep{caragiannis2016unreasonable}. 

For these reasons, CEEI and its variants are used in real world allocation problems such as allocating courses to students in business schools \citep{budish2011combinatorial} and the splitting of tasks such as chores between roommates (Spliddit.com, \cite{goldman2015spliddit}). 

By contrast to the fair division literature, the algorithmic fairness literature focuses on `fairness' of treatment \textit{across groups} (see the Related Work section). We focus on two particular notion of fairness. The first, motivated by real world constraints: a mechanism should not have disparate impact \textit{in outcomes} across protected classes \citep{barocas2016big}. This is particularly relevant when the `items' in question are those protected by law e.g. job interviews.\footnote{Note that we do not claim this is always that this notion of algorithmic fairness is the relevant one - this will depend on the particular application and circumstances.} We also focus on a different sense of disparate impact. We allow allocations to differ across groups but we require that \textit{realized utility} is equalized across groups (on average). 

CEEI may generate allocations which do not respect either definition of disparate impact. Our contributions are to present two modifications to the CEEI procedure. The first, which we refer to as Equitable Equilibrium from Equal Incomes (EqEEI) generates allocations without disparate impact in outcomes. To implement EqEEI we use a technique based on kernel two-sample test using the maximum mean discrepancy distance \citep{fortet1953convergence,gretton2012kernel} to `pre-process' a set of valuations. We then use these processed valuations in a modified CEEI. We show this yields allocations whose distributions do not differ across protected classes.

The second generates allocations without disparate impact in realized utilities. We show this can be done by changing the budget of the disadvantaged group in a CEEI-like procedure. We refer to as Competitive Equilibrium from Equitable Incomes (CEEqI).

We study EqEEI and CEEqI analytically. We show that removing disparate impact in allocations carries costs in terms of good properties of the mechanism - EqEEI has envy, regret, is not Pareto efficient, and loses incentive compatibility. We show how to bound these losses by invoking recent results on market abstractions \cite{kroer2019computing}. By contrast, we show that CEEqI remains a no regret, (budget adjusted) envy free, Pareto efficient, and incentive compatible mechanism. Thus, we argue, it may be a useful mechanism for trading off distributional and group-level fairness constraints in practice. Finally, we evaluate the tradeoffs from both EqEEI and CEEqI in a real dataset.

\section{Related Work}
CEEI is far from the only mechanism for fair division and finite supply linear utilities are far from the only valuation function. Another commonly studied problem is the allocation of a divisible single good (often called a cake) to agents with preferences over subsets of the cake \citep{brams1996fair,chen2013truth}. This is more complicated than our setting and future work should investigate the use of various algorithmic fairness constraints to the problem of fair division in this setting.

Other work in fair division has looked at group level metrics for example by expanding notions such as envy-free and maximin shares to group equivalents \citep{aleksandrov2018group,conitzer2019group}. The closest such extension to our work is the concept of group-envy-freeness either with respect to existing groups \citep{aleksandrov2018group} or a much stronger version which extends to \textit{all groups at once} \citep{conitzer2019group}. 

Our choice of disparate impact as our desideratum is a choice motivated by practical considerations and is not be applicable in all environments. A large literature on algorithmic fairness discusses many ways that fairness can be defined and which kind of fairness is appropriate for which kind of situation \citep{dwork2012fairness,barocas2016big,berk2018fairness,kusner2017counterfactual}. A set of papers also shows that different notions are incompatible and so real world systems must trade off between these different definitions \citep{kleinberg2016inherent,corbett2017algorithmic,friedler2016possibility}. An interesting extension to our preliminary results would be to incorporate other forms of algorithmic fairness into allocation mechanisms so that practitioners have multiple, off-the-shelf approaches that can be used in the appropriate situations. 

Our approach of constructing representations that are invariant to the protected class has been used in applying algorithmic fairness notions to the classification problem, e.g. a classifier for which individual should get a loan \citep{calders2009building,feldman2015certifying,zemel2013learning,louizos2015variational}. Our approach is also strongly related to recent work which shows that collaborative filtering \citep{yao2017beyond} or unsupervised learning such as word embeddings can capture stereotypes and biases that are present in the datasets and presents techniques for removing such biases \citep{bolukbasi2016man}. Our main contribution is to take these ideas and apply them to the problem of allocation, we do not claim to originate debiasing methods in general.

\section{Competitive Equilibrium from Equal Incomes}
We consider a set of $n$ individuals $\mathcal{U}$ with generic individual $i$. There are also a set of $m$ divisible items $\mathcal{J}$ with generic item $j$. Items have finite supply which here we normalize to $1$.

\begin{defn}
An \textbf{allocation} is an $x \in \mathbb{R}^{n \times m}_{+}$ with $x_{ij}$ indicating the amount of item $j$ that has been allocated to individual $i$. Let $X$ be the set of possible allocations given the supply constraints with $X_i$ being allocations that only allocate to individual $i$.
\end{defn}

Each individual $i$ has a valuation for each item $j$ written as $v_{ij}$. We assume valuations are linear so the total utility to person $i$ for an allocation $x$ is given by the standard dot product $v_i \cdot x_i = \sum_{j=1}^m v_{ij}x_{ij}$. Valuations are strictly positive and there exists a maximum valuation $\bar{v}.$

We focus on a particular allocation mechanism: competitive equilibrium from equal incomes (CEEI). The mechanism works as follows: each individual is given a unit budget (though generically this can be a budget of size $B_i$) of pseudo-currency. Each item $j$ is given a price $p_j$. Given these prices and a budget, the set of affordable allocations for individual $i$ is $\{x_i \in X_i : p \cdot x_i \leq B_i \}.$

\begin{defn}
  An individual $i$'s \textbf{demand} given prices $p$ and budget $B_i$ is
  \[
    d_i(p, B_i) = \lbrace x_i \in X_i \mid v_i \cdot x_i \geq v_i \cdot x'_i\ \forall x'_i \in X_i \text{ s.t. } x_i'\cdot p \leq B_i \rbrace.
  \]
  Note that the demand can be set valued but the achieved utility level of the demand, which we write as $\bar{d}_i (p, B_i)$ is unique.
\end{defn}

Given individual demand functions we can define an equilibrium.

\begin{defn}
A \textbf{market equilibrium} is a feasible allocation $x^* \in X$ and a set of prices $p^*$ such that for all $i$ $x^*_i \in d_i (p, B_i)$ and for all items $j$ $\sum_i x_{ij} = 1.$
\end{defn}

A particular market may have multiple equilibria. However, there is one market equilibrium that is always guaranteed to exist: the Eisenberg-Gale equilibrium \citep{eisenberg1959consensus}, where each individual gets the same bang-per-buck across across all items assigned to the individual. We use this specific equilibrium, and refer to the assignment mechanism which uses this equilibrium as CEEI\footnote{We refer to \text{the} CEEI across this text, though this is not technically correct: CEEI prices are unique but there may be multiple allocations $x$ corresponding to the unique prices $p$ when individuals have ties across items. However, the utility levels attained across these allocations are equal.}. It is well known that this equilibrium (and many related variants) can be found by solving a convex program \citep{eisenberg1959consensus,cole2017convex}:

\begin{align*}
\max_{x} \sum_i B_i \text{log} (v_i \cdot x_i) \text{ subject to } \forall j \sum_{i} x_{ij} \leq 1
\end{align*}


The convex program interpretation means the outcomes of our mechanism will be distributionally fair in the sense of maximizing the product of utilities. In addition, because our allocation is a market equilibrium we are guaranteed envy freeness, no regret (relative to the prices), and that each individual gets at least their maximin share. In addition, the CEEI prices are unique though allocations may have ties and so may not be.

A nice structural feature of the solution to the CEEI optimization problem is that it is the maximal point of the set of \emph{utility prices} for individuals. The utility price of an individual $i$ is the amount of budget they need to spend to acquire $1$ point of utility (the inverse of this is often referred as the bang-per-buck rate of individual $i$). It is easily seen that in the optimal solution each individual has a single utility price $\beta_i = \frac{B_i}{v_i\cdot x_i}=\frac{p_j}{v_{ij}}$ for all $j$ such that $x_{ij}>0$. Thus we may think of the algorithmic problem of finding the CEEI solution as one of finding the optimal utility prices for each individual.\footnote{This observation regarding the maximality properties of utility prices was previously proven in the quasi-linear case (where utility is given by $x_i \dot v_i - p \cdot x_i$) by \citet{conitzer2019pacing} and in the Appendix we extend this result to the our case.}

\subsection{Large Market Results}
An important part of CEEI is individuals reporting their valuations to the mechanism. A mechanism that ensures individuals report their valuations truthfully is known as \textit{strategy proof}. While it is known that individuals can often misreport their valuations and get better outcomes in CEEI \citep{branzei2014fisher} it is also known that many non-strategy proof mechanisms  can become \emph{strategy proof in the large} (SP-L).

We will be mostly interested in large scale allocation problems (e.g. courses to students), thus here we turn to a stochastic model to study whether CEEI (and our later proposed mechanisms) are indeed SP-L.\footnote{Our results are stronger than existing results \citep{jackson1997approximately,azevedo2018strategy} as we give a convergence rate for large finite markets rather than considering limit definitions for continuous markets. In addition, our results extend the general results in \citet{azevedo2018strategy} as we allow for a continuous outcome space.}

We will use the following stochastic model for growing our market. There are $n$ randomly drawn individuals and a constant set of $m$ items. The supply of each item grows linearly in market size and is given by $s_j=c_j n$. There is a distribution of valuations $\mathcal{F}$ which has full support on the set of possible valuations $[0, \bar{v}]^m.$

We consider the definition of SP-L introduced by \citet{azevedo2018strategy}. Let $\sigma$ be a mapping (not necessarily truthful) from valuations to reports. Suppose that $\sigma$ has full range (i.e. for any report, there is a type that gives that report with some probability). Consider an individual $i$ who can either report their true valuations $v_i$ or some other valuation $v'_i$ with everyone else reporting according to $\sigma$. Consider $i$'s expected utility with the expectation taken relative to the $n-1$ other individuals behaving according to $\sigma$ with their true types drawn from $\mathcal{F}$.

\begin{defn}
 A mechanism is SP-L if, given $\sigma$ and any $\epsilon > 0$ there exists a $\bar{n}$ such that if $n > \bar{n}$ then the gain from reporting any $v'_i$ instead of their true valuation $v_i$ gains is at most $\epsilon$.
\end{defn}

As \citet{azevedo2018strategy} point out, this notion is stronger than dominant strategy incentive compatibility but weaker than simply requiring that truth-telling be a Bayesian Nash equilibrium. We first show a result about CEEI which we can then use in later proofs. 

\begin{theorem}
  For any $\sigma$ the expected utility gain from misreporting for an individual $i$ in the CEEI mechanism is bounded by $O(\frac{1}{n})$. Thus CEEI is SP-L.
  \label{thm:ceei spl}
\end{theorem}

Our proof consists of a series of lemmas about the effect on the market that any individual agent can have. The lemmas may be of independent theoretical interest but in the interest of space we relegate them to the appendix. First we show that the amount by which an individual can affect prices decreases with market size. Second, we show that for the proposed stochastic model market prices become bounded from below. This allows us to give the convergence rate in the Theorem.
%
%

\section{Disparate Impact}
While CEEI is distributionally fair, it may be unfair in a different way. Suppose each individual is associated with a binary \textit{protected class} label $z_i$. The literature on algorithmic fairness focuses on disparate impact across these protected classes. Disparate impact may occur in one of two ways. First, mechanism designers may care about overall allocation of items not differing across groups. This is particularly relevant when mechanisms are used in legally protected domains such as the allocation of items related to credit, housing, or jobs. Second, designers may care about realized \textit{utility levels} being equalized across groups. The CEEI allocations may not respect either definition of fairness.


\subsection{Equitable Equilibrium from Equal Incomes}
Given an allocation $x$ let $f_{x}^z$ be the probability mass function for the empirical distribution of allocations for individuals with protected class $z \in \{0,1\}$. Our first notion of algorithmic fairness requires that the  allocation distribution conditioned on the protected class is the same across classes.

\begin{defn}
We say that an allocation \textbf{has no disparate impact in allocations} if $f_{x}^1 = f_{x}^0.$
\end{defn}

It is fairly clear that CEEI may not respect the protected class. Even though CEEI does not use $z$ directly, $v$ may be related to $z$ and thus lead to allocations that do not respect the protected class.

Unfortunately, the simple solution of adding $f_x^1=f_x^0$ as a constraint in the convex program above changes the problem away from a convex program and the solution becomes a MIP. However, it is well known that (in the case of random variables) if two variables are independent, then functions of them are also independent.

This suggests the following procedure which we refer to as \textit{equitable equilibrium from equal incomes} (EqEEI). Everyone reports their valuations, we construct a matrix $V$. We transform the matrix into $\hat{V}$ such that $f_{\hat{V}}^0 = f_{\bar{V}}^1$, where $f_{\hat V}^z$ is the empirical distribution of $\hat V$ under protected class $z$. Because of this procedure's similarity to past work we refer to this as the set of `debiased' valuations.

We then compute CEEI using $\hat{V}.$ Note that for $f_{\hat{V}}^0 = f_{\bar{V}}^1$ to hold, it must be that for any individual's valuation vector $\hat{v}$, it is equally likely to be found in class $0$ or $1$ - when the number of individuals is finite, this means that if there is an individual in class $0$ with a vector $\hat{v}$ there must also be one in class $1$ with the same $\hat{v}.$

We can pool all items assigned to individuals with vector $\hat{v}$ and split each item uniformly among them.\footnote{This is equivalent to creating a representative buyer in the abstraction framework of \citet{kroer2019computing}. Creating a representative buyer may be preferable for computational purposes.} Because these individuals all have vector $\hat{v}$ they are indifferent among the pooled set of items (under $\hat V$), and any reassignment is also guaranteed to be a CEEI under $\hat V$.

Because allocations are now deterministic functions of $\hat{v}$ vectors and because $\hat{v}$ does not distinguish classes we can formally state that:

\begin{theorem}
  An EqEEI $(p, x)$ has no disparate impact in allocations.
  \label{thm:eqeei no disparate impact}
\end{theorem}

Of course, one choice of $\hat{V}$ is setting $\hat v_{ij}$ to a constant and giving everyone an equal fractional share. However, this will give a very poor allocation. Instead we would like to choose $\hat V$ such that we somehow maintain `good' allocation properties. To this end, our choice of $\hat{V}$ can be guided by known results:

\begin{knownresult}[\cite{kroer2019computing}]
  If we compute the CEEI allocation $x$ for valuation matrix $\hat{V}$ and use $x$ with the real valuation matrix $V$ then quantities such as envy, regret, and maximin share, can be bounded by a linear function of $ \| V - \hat{V} \|_{1,\infty}.$
  Finally the Nash social welfare can be bounded by a multiplicative function of $\| V - \hat{V}\|_{1,\infty}.$
\label{known:kroer19}
\end{knownresult}

The $1, \infty$ norm is difficult to minimize, so \citet{kroer2019computing} suggest using the Frobenius norm instead. Thus, to apply EqEEI in practice, we propose to solve 
$
\min_{\hat{V}} || V - \hat{V} ||_{F} \text{ subject to } f_{\hat{V}}^0 = f_{\hat{V}}^1
$.


Then given this solution, we will compute the CEEI allocation using $\hat{V}$. EqEEI will then be a CEEI with respect to $\hat{V}$.


Next we show a counterexample showing that EqEEI is unfortunately not strategyproof, even in the large. In other words, when EqEEI is deployed, individuals can gain non-negligible utility from misreporting their valuation vector. The reason is that when the valuation distributions $f_V^0,f_V^1$ conditioned on $z$ are very different from each other, then some individuals must be mapped to a debiased valuation that is quite far from their original, causing them to buy suboptimal bundles. 

\begin{example}
We will consider a counterexample consisting of a market with a continuum of individuals. Let there be two items: item $A$ has valuation 1 for every individual. Item $B$ is where heterogeneity occurs and individuals have a valuation in $v \in [0,2]$ for this item.
 
The continuous set of individuals is parametrized by pdfs $f^1, f^0 : [0,2] \rightarrow \mathbb{R}_+$ for $z=0$ and $z=1$ respectively, with $f^i(z)$ being the density of individuals of class $i$ with valuation $z$ for item $A$. In the continuous case the supply constraint is that $\int_i x_{ij} = 1$ for $j=A,B$.

Let the median values of $f^0$ and $f^1$ be $\bar{B}^0$ and $\bar{B}^1$, with $\bar{B}^0 < 1 < \bar{B}^1$. In a standard CEEI more individuals from class $1$ will receive item $B$ and more individuals from class $0$ will receive item $A$. Thus, CEEI will not respect disparate impact with respect to allocations. In this case `debiasing' maps each individual in each class to a new valuation $\hat{v} (z, v)$ such that the resulting distribution from both classes is $\hat{f}$. Let us consider the case where the $\hat{f}$ continues to have full support, this is not required but makes the logic clearer.

CEEI with respect to this new $\hat{f}$ generates an allocation map $\hat{x} (\hat{v}).$ Notice that these allocations can only include a measure $0$ set of types that receive both items. This is because CEEI sets equal rates and if prices are such that some $\bar{\hat{v}}$ receives equal rates then any $\hat{v} > \bar{\hat{v}}$ must prefer a pure allocation of item $B$ at these prices while any $\hat{v} < \bar{\hat{v}}$ would prefer a pure allocation of item $A$. 

However, since $f^1$ and $f^0$ are mapped to the same $\hat{f}$, one of the following must be true: there is a $v$ in class $1$ with $\hat{v} (v) < \bar{\hat{v}}$ but $v > \bar{\hat{v}}$  or there is a $v$ in class $0$ with $\hat{v}(v) > \bar{\hat{v}}$ but $v < \bar{\hat{v}}.$ Thus, these individuals will be receiving the `wrong item' with respect to their true valuations. Thus, they would like to misreport and pretend to be a different type $v'$ in their class which is receiving the correct item. Such a $v'$ must exist since allocations are not distinguishable across classes and so individuals in both protected classes must be receiving some allocation of purely a single item.
\end{example}

While EqEEI is not strategyproof in the large, it does give a very strong guarantee: allocations are indistinguishable based on $z$. Thus it may be a desirable solution concept when such guarantees are required, e.g. by law. Next we will investigate a solution concept that is allowed to distinguish on $z$, but where we get stronger guarantees in terms of properties satisfied by standard CEEI.

\subsection{Competitive Equilibrium from Equitable Incomes}
We now consider equitable allocations in utility space. Given an allocation $x$ let $U(x)_z$ be the average utility of individuals with protected class $z$ in allocation $x$.

\begin{defn}
We say that an allocation $x$ \textbf{has no disparate impact in utilities} if $U(x)_1 = U(x)_0.$
\end{defn}

It is clear that vanilla CEEI may not satisfy the condition above. For simplicity, we assume that group $0$ is the group which under $x^{CEEI}$ has $U(x)_1 < U(x)_0.$  

We propose the following mechanism: instead of everyone being allocated an equal budget of size $1$, everyone with protected class $z=0$ is allocated a budget of $1$ and everyone with $z=1$ is allocated a budget of $B_1 > 1$ such that average utilities equalize. Such a $B_1$ is guaranteed to exist:

\begin{theorem}
There exists $\bar{B}$ such that the EG equilibrium allocation $x^*_{\bar{B}}$ has $U(x^*)_1 = U(x^*)_0.$ 
\end{theorem}

The proof of the theorem follows from the following. First, for $B_1 = 0$, we have $U(x_{0}^*)_1 < U(x^*_0)_0.$ Second, for large enough $B_1$, we have $U(x_{0}^*)_1 > U(x^*_0)_0.$ Third, the EG allocations $x^*$ are continuous in $B$. Therefore, by the intermediate value theorem, there exists $\bar{B}$ such that $U(x^*)_1 = U(x^*)_0.$ 

We refer to the mechanism which uses the EG allocation from this $\bar{B}$ as \textit{competitive equilibrium from equitable incomes}. Importantly, this allocation can be computed efficiently.

\begin{corr}
For any $\epsilon > 0$ we can compute $\bar{B}_{\epsilon}$ such that the EG equilibrium computed using $\bar{B}_{\epsilon}$ has $| U(x_{0}^*)_1 - U(x^*_0)_0 | < \epsilon$ using $\text{log} (\frac{1}{\epsilon})$ number of CEEI solves.
\end{corr}

Since the difference $U(x^*)_1 - U(x^*)_0$ is monotone we can even use a simply binary search procedure to find $\bar{B}.$ Once we have this $\bar{B}$ we can use the EG equilibrium allocation as the item assignments. We refer to this procedure as \textit{competitive equilibrium from equitable incomes} (CEEqI). In what follows we assume WLOG that $B_1 \geq 1$.

Unlike EqEEI, CEEqI will preserve many properties of CEEI quite well. Indeed, we can show that:
\begin{theorem}
  CEEqI has no regret (individuals prefer their CEEqI allocations to any allocation they could afford given their CEEqI budgets), no scaled envy (every individual in the group with larger budget has no envy, whereas individuals in the group with lower budgets has no envy after scaling their utility by $B_1$), CEEqI allocations are Pareto optimal, CEEqI is SP-L.
  \label{thm:ceeqi sp-l}
\end{theorem}

The first three properties follow directly from CEEI being a market equilibrium. For SP-L we relegate the proof to the appendix; intuitively it is an application of Theorem~\ref{thm:ceei spl} (which tells us that individuals have negligible impact on prices as markets get large) with the additional note that as the market gets large the effect of any individual on the chosen value of $B_1$ goes to zero - thus both $B_1$ and prices become exogenous to any individual $i$ and reporting true valuations becomes utility maximizing.


\section{Experiment}

Thus, we now have everything required to perform a real experiment. We take a classic recommendation systems dataset, MovieLens, which has $1$ million ratings of $\sim 4000$ movies by $\sim 6000$ users. 

The rating data does not include all ratings for all movies. We use a standard technique from recommender systems to construct a low rank approximation to the full valuation matrix (see Appendix).

MovieLens contains metadata about the users rating movies. In particular we can look at user reported gender (dichotomized into male and female). We train a simple $2$ layer neural network to predict gender from the learned $U$ vectors and find we can achieve an out-of-sample AUC of approximately $\sim .76$. Thus, if we were to construct an allocation system based on these vectors, even though such a system does not \textit{explicitly} use gender as an input, its outputs may be quite different on gender lines. 

To further underscore this point, remember that the vector representations trained above embed users and movies into the same vector space. Thus, if we have a function $\hat{g}(u_i)$ which inputs a user vector and outputs a guess about the user's gender we can take a trained movie vector $m_j$ and pass it through $\hat{g}$ to get an estimate of the `genderedness' of the movie. 

In figure~\ref{fig:stereotypes}  we show the `most male' and `most female' movies according to this procedure. We see that indeed strong gender stereotypes are encoded into the movie vectors. 

\begin{figure*}[!ht]
  \vspace{-3mm}
\begin{center}
  \scriptsize
\begin{tabularx}{\textwidth}{| X | X | }
\hline
\textbf{Male Stereotyped Movies} & \textbf{Female Stereotyped Movies} \\
\hline
Half Baked (1998), Dumb and Dumber (1994), South Park: Bigger, Longer and Uncut (1999), Beavis and Butt-head Do America (1996), Happy Gilmore (1996), Evil Dead II (Dead By Dawn) (1987), Texas Chainsaw Massacre, The (1974), Aliens (1986), From Dusk Till Dawn (1996), Me, Myself and Irene (2000), Halloween (1978), Rocky IV (1985), Kingpin (1996), Waterboy, The (1998), Billy Madison (1995), Plan 9 from Outer Space (1958), Godzilla 2000 (1999), Dirty Work (1998), Conan the Barbarian (1982), Night of the Creeps (1986) & 
 Walk in the Clouds, A (1995), Mulan (1998), Parent Trap, The (1998), Ever After: A Cinderella Story (1998), Anastasia (1997), Out of Africa (1985), Color Purple, The (1985), Ghost (1990), Tarzan (1999), Mary Poppins (1964), Sense and Sensibility (1995), Titanic (1997), Little Mermaid, The (1989), Beauty and the Beast (1991), Pretty Woman (1990), Newsies (1992), Grease (1978), Dirty Dancing (1987), Sound of Music, The (1965), Condorman (1981) 
 \\
 \hline
 \end{tabularx}
 \scriptsize
\end{center}
\caption{Embedding systems trained on past ratings produce vectors that contain gender stereotypes. 
}
\label{fig:stereotypes}
  \vspace{-3mm}
\end{figure*}

Perhaps in the case of allocating scarce movie theater seats these biases would not be a big deal (or maybe they would). However, if we were to think of a similar system being used for e.g. allocations of job ads, then things start to look different. Problems can also arise if the system is allocating goods to individuals where an adversary may want to use knowledge of an individual's allocation to try to reverse engineer a \textit{hidden} protected class membership (e.g. membership in a vulnerable group).

From the trained model we construct the complete matrix of valuations. We choose a submatrix of the top $1000$ most rated movies and top $1000$ users with the most ratings to construct $V$ for our market. We compute the various equilibria for our mechanisms using the EG convex program where the movies are items, by default have budget $1$, and all movies have supply $1$. 

\subsection{EqEEI Allocations}
In order to construct the EqEEI allocation we need to construct the `debiased' matrix $\hat{V}.$ To do so we parametrize $\hat{V} = f(V)$ via a neural network. To evaluate the constraint $\bar{\hat{v}}_1 = \bar{\hat{v}}_0$ we can express this via any distance between distributions. We choose the maximum mean discrepancy distance (MMD, \cite{fortet1953convergence}). The MMD is used in the literature on testing whether two probability distributions are equal from samples \citep{gretton2009fast,gretton2012kernel,lopez2016revisiting}.
\cite{gretton2012kernel} shows that in the case of two random probability distributions $X$ and $Y$ we can test if they are equal by taking a well behaved kernel $k$ (we use the Gaussian kernel) and computing $$\text{MMD}(X, Y) = \mathbb{E} \left[ k (X, X) \right] - 2 \mathbb{E} \left[ k(X, Y) \right] + \mathbb{E} \left[k(Y, Y) \right].$$ The MMD here will be $0$ iff $X=Y$. In finite samples, we simply take sample means from the distributions instead of the expectation. 

The literature on two-sample testing focuses on asking whether two distributions are equal (and so derives test statistics for finite sample applications of the MMD). We will instead use the MMD to make them so using a continuous relaxation of the constrained optimization: $$\max_{f} || V - f(V) ||_2 + \lambda \text{MMD}(\bar{\hat{v}}_1, \bar{\hat{v}}_0).$$  This formulation is non-convex (as $f$ will be a neural network) but in practice we find that it can be solved satisfactorily via stochastic gradient descent.

We apply MMD-based pre-processing to the user level vectors $U$. By setting $\lambda$ high enough (in our case, $2e6$) we can reduce the ability of a classifier to detect gender from the vectors. We find that when $U$ is transformed the AUC of a trained classifier on $f(V)$ is $0.5$, i.e. at chance levels. 

We now look at various losses in allocation quality from using EqEEI instead of CEEI. Let $x^*$ be the CEEI allocation for the original matrix $V$ and let $(\hat{x}^*, \hat{p}^*)$ be the EqEEI allocation (i.e. the one generated from CEEI on $\hat{V}$). We consider how enforcing a lack of disparate outcomes affects the following metrics, which are normalized such that each metric shows the fraction of the optimal value missing or obtained, depending on metric.
\begin{center}
\vspace{-2mm}
  \small
\setlength\tabcolsep{3.5pt}
\begin{tabular}{lll}  
  $\text{Regret}_i (\hat{x}_i, \hat{p})$ &=& $\dfrac{v_i \cdot \hat{x}_i - \bar{d}_i (\hat{p}, b_i) }{\bar{d}_i (\hat{p}, b_i)} $         \\
  \midrule
  $\text{Envy}_i (x)$ &=& $\dfrac{\max_{i'} v_i \cdot (x_i - x_i') }{\max_{i'} v_i \cdot x_i'}$      \\
\end{tabular}
\quad
\begin{tabular}{lll}  
$\text{Pareto Gap}(\hat{x})$ &=& $\dfrac{\sum_i v_i \cdot \hat{x}_i}{\sum_i v_i \cdot x^{P} (\hat{x})}$            \\
  \midrule
  $\text{Geometric Mean Gap}$ &=& $\dfrac{\text{exp}({\sum_i \text{log} (v_i \cdot \hat{x}_i)})}{\text{exp}({\sum_i \text{log} (v_i \cdot x^*)})}$         \\
  \midrule
  $\text{Efficiency Gap}(x^*_i, \hat{x}_i)$ &=& $\dfrac{\sum_i \hat{x}_i \cdot v_i}{\sum_{i} x^*_i \cdot v_i}$      \\
\end{tabular} \\
\vspace{-0.5mm}
\end{center}

We have two metrics that are for each individual $i$. $\text{Regret}_i (\hat{x}_i, \hat{p})$ is how close $i$ is to achieving the utility of their demand.  $\text{Envy}_i (x)$ is how much $i$ would prefer the allocation of some other individual. Recall that in true CEEI regret and envy is $0$.

We have three `market level' metrics. A market equilibrium is guaranteed to be Pareto efficient - nobody can be made better off without making someone else worse off. This may not be preserved in EqEEI however. $\text{Pareto Gap}(\hat{x})$ is the social welfare achieved at the solution $\hat x$ divided by the social welfare at the best Pareto-improving allocation $x^P(\hat x)$. $\text{Geometric Mean Gap}$ is the geometric mean of utilities achieved in the solution $\hat x$ divided by the geometric mean of utilities in the optimal solution $x^*$ to the original problem. Finally,  $\text{Efficiency Gap}(x^*_i, \hat{x}_i)$ is what fraction of the social welfare from the original CEEI solution is obtained in EqEEI or CEEqI.

In our MovieLens constructed market we see that each metric has a modest $5$-$10\%$ loss from using EqEEI over CEEI:
\begin{center}
  \vspace{-3mm}
  \small
\begin{tabular}{ c |  ccccc} 
   Metric & Regret & Envy & Pareto Gap & Geometric Mean Gap & Efficiency Gap \\
   \hline
   EqEEI loss & -0.1 & -0.08 & 0.95 & 0.94 & 0.94 \\
\end{tabular}
  \vspace{-3mm}
\end{center}


\subsection{CEEqI Allocations}
We now study CEEqI allocations. When we compute CEEI for the $1000$ by $1000$ market described above, we find only a minor imbalance between class utilities (on the order of $5 \%$). To really study the tradeoffs between class-level distributional fairness and other objectives, we construct an instance with larger class imbalance. To do so, we construct a new valuation matrix by multiplying the $v_i$ vectors of all individuals in class $1$ by $.75$.  

This yields CEEI outcomes with substantial disparate impact in utilities (on the order of $\sim 25 \%$). We now consider using CEEqI to remedy this issue. Recall that CEEqI only changes budgets - therefore it generates Pareto optimal, no regret, no (scaled) envy allocations. For evaluation, we only consider the geometric mean gap and the inequity as a function of budget for the protected group (the other group receiving always a budget of $1$). 

Figure \ref{fig:CEEqI} shows that changing budgets equalizes average utilities while only changing the geometric mean utility attained by the equilibrium a little bit. Intuitively this happens for several reasons. First, the class 1 only forms about $\sim 30 \%$ of the market. Second, the item transfer is efficient in the sense that we know from examples above that there are items that are strongly preferred by one class. This means that we can increase the average utility of class $1$ by moving those items from class $0$ at a relatively small cost (since class $0$ individuals do not actually value those items very highly). 

\begin{figure}[h!]
\begin{center}
  \begin{minipage}[t]{0.57\textwidth}
\vspace{0pt}\hfill
 \includegraphics[scale=.4]{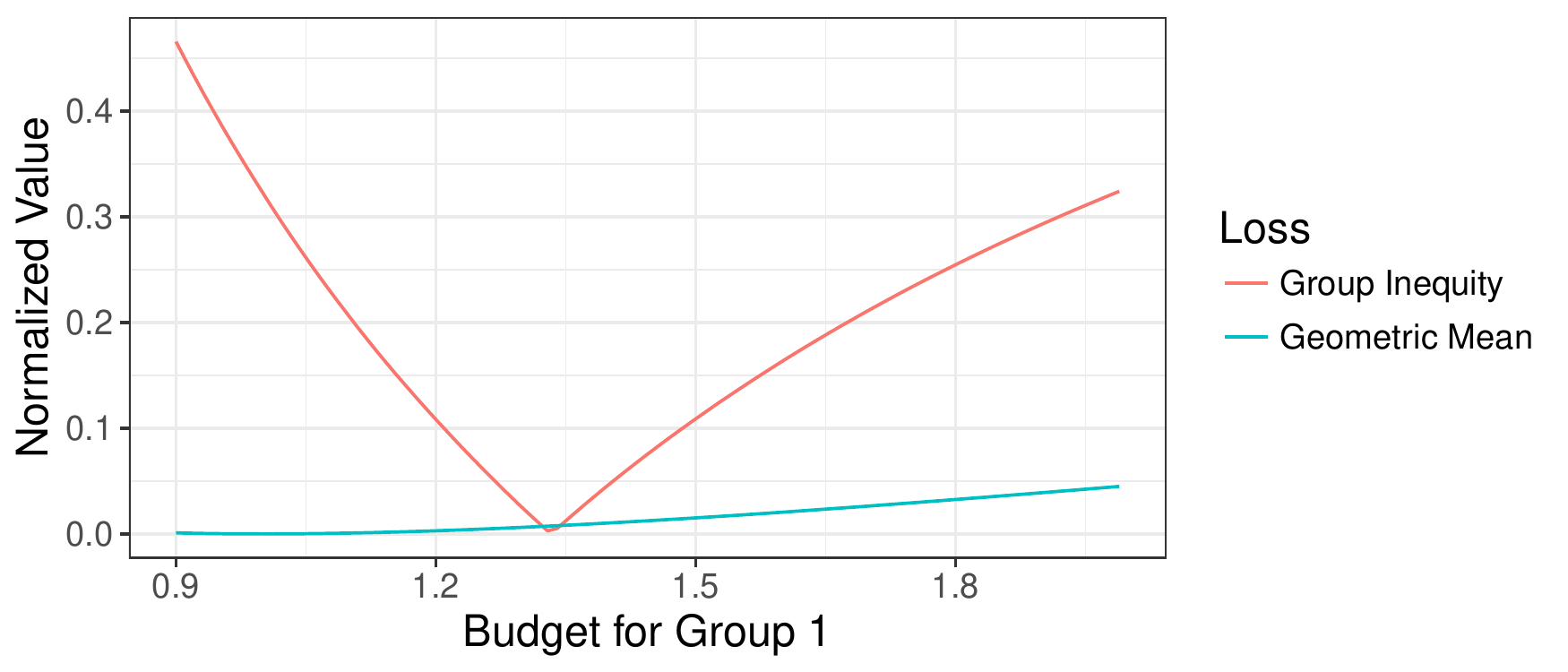}
  \end{minipage}
  \begin{minipage}[t]{0.35\textwidth}
\vspace{0pt}
    \caption{ CEEqI results. Even after modifying the valuation matrix to increase group level disparate outcomes in utility we see that CEEqI can reduce this difference at only a small cost to the overall objective of the geometric mean.
    } \label{fig:CEEqI}
    \vfill\hfill
  \end{minipage}
\vspace{-7mm}
\end{center}
\end{figure}

\section{Conclusion}
The notion of fairness of treatment has become an important research area in machine learning. Mechanism design is a practical field. We do not claim to have solved all problems of `unbiased' mechanisms in this paper. Rather, we offer mechanism designers a particular tool that can be useful in some circumstances. In addition, we hope that this paper can initiate a larger conversation about the role of fairness in various forms in mechanism design.


\bibliographystyle{ACM-Reference-Format}
\bibliography{main_fair_unbiased.bbl}

\clearpage
\appendix
\section{Utility Prices and CEEI}
The rest of the section assumes $s_j=1$ for all $j$, which is without loss of generality. The set of utility prices that we are interested in are those such that they correspond to feasible allocations and prices in the following sense
\begin{defn}
  A vector $\beta \in \mathbb{R}_{++}^n$ of \emph{utility prices} are budget feasible if there exists $x,p$ such that $p_j = \max_i \beta_i v_{ij}$ for all $j$, $x_{ij} > 0 \Rightarrow i \in \argmax_{i'} \beta_{i'}v_{i'j}$, and $p\cdot x_i \leq B_i$ for all $i$.
\end{defn}

We now show that taking the elementwise max over a pair of utility-price vectors preserves budget feasibility. We then show that CEEI corresponds to the unique maximal point of the set of budget-feasible utility prices.

\begin{lemma}[Elementwise max preserves feasibility]
  Let $\beta, \beta'$ be budget-feasible utility prices. Then their elementwise max $\beta^+ = \max (\beta, \beta')$ is also budget feasible.
\end{lemma}
\begin{proof}
  Since $\beta$ is budget feasible there exists $x,p$ such that $p_j = \max_i \beta_i v_{ij}$ for all $j$ and $p\cdot x_i \leq B_i$ for all $i$. Similarly let $x',p'$ be such a pair for $\beta'$. Now let $p^+ = \max_i \beta^+ v_{ij}$. We construct the allocation $x^+$ as follows: for all $i$ let $\hat x_i, \hat p^i$ be the allocation and price vectors corresponding to the allocation of $i$ under the max of $\beta_i, \beta_i'$, we set $x^+_{ij} = \hat x_{ij}$ for all $j$ such that $\beta^+_iv_{ij} = \max_{i'} \beta^+_{i'}v_{i'j}$, and we set $x^+_{ij} = 0$ otherwise. Now for any $j$ such that $x^+_{ij}>0$ we have $p^+_j = p^i_j$, and for all $j$ we have $x^+_{ij} \leq \hat x_{ij}$ since $\max_{i'}\beta^+_{i'} v_{i'j}$ weakly increased relative to both $\beta$ and $\beta'$. It follows that $p^+ \cdot x^+_i \leq B_i$.
\end{proof}

\begin{lemma}[CEEI corresponds to maximal $\beta$]
  The CEEI solution corresponds to the unique maximal utility prices $\beta$ such that $\beta \geq \beta'$ for all $\beta'$ that are budget feasible.
  \label{lem:ceei maximal utility prices}
\end{lemma}
\begin{proof}
  First we note that $\sum_j p_j = \sum_i B_i$ in CEEI. Thus, any non-maximal $\beta$ cannot be a CEEI solution, since there must then exist $\beta' \geq \beta, \beta' \ne \beta$ which is budget feasible. But then the prices $p'$ under $\beta'$ are still budget feasible and $\sum_{j} p_j' > \sum_j p_j$; this shows that $\sum_j p_j < \sum_i B_i$ for all non-maximal $\beta$. Furthermore, a maximal $\beta$ is guaranteed to exist, as we know that $\beta \in (0, \max_i \frac{\sum_{i'} B_{i'}}{\sum_j v_{ij}})$ and thus the set of possible choices for $\beta$ is closed and bounded.
  Since the CEEI solution is guaranteed to exist it must correspond to the unique maximal $\beta$.
\end{proof}

\section{Proof that CEEI is SPL}
Our proof relies on the maximality structure of utility prices shown in Lemma~\ref{lem:ceei maximal utility prices} .



\begin{lemma}
  For any market, an individual can affect the price $p_j$ of any item $j$ by at most $\frac{B_i}{s_j}$.
  \label{lem:small price effect}
\end{lemma}
\begin{proof}
  Let $i$ be an arbitrary individual. Now consider the CEEI solution $(x,p)$ when $i$ does not participate in the market, and let $\beta \in \mathbb{R}^{n}$ be the associated per-individual utility prices, with $\beta_i=0$. We know that $(\beta, x, p)$ constitutes a budget-feasible solution to the market that includes $i$. Now, by the monotonicity of utility prices we know that the CEEI for the market that includes $i$ must have utility prices $\beta' \geq \beta$.
  It follows from the utility-rate monotonicity that prices must go up, but we also know that the new prices $p'$ must satisfy $\sum_j s_jp_j' = \sum_{i'} B_i$, whereas $\sum_j s_jp_j = \sum_{i' \ne i} B_i$. It follows that, at worst, individual $i$ can increase the supply-weighted price $s_jp_j'$ by at most $B_i$, since otherwise $\sum_j s_jp_j' > \sum_{i'}B_i$.
  It follows that no matter what value individual $i$ reports, $p_j' \in \left[p_j, p_j + \frac{B_i}{s_j}\right]$.
\end{proof}
If we were not interested in obtaining a rate of convergence to SP-L, then Lemma~\ref{lem:small price effect} could be combined with the results of \citet{roberts1976incentives,jackson1997approximately} to show SP-L. A similar rate-less result follows from \citet{roberts1976incentives}.


\begin{lemma}
  For all $j \in \mathcal{J}$ and any set of individuals $\mathcal{I}'$ such that some $i$ in $\mathcal{I}'$ has $v_{ij} > 0$, we have
  \[
    p_j \geq \frac{|\mathcal{I}'|\min_iB_i}{n\sum_{j' \in \mathcal{J}}c_{j'} v_j''}v_j' \eqdef p_j^{\downarrow},
  \]
  where $v_j' \leq \min_{i\in \mathcal{I}'}v_{ij}$ and $v_{j'}'' \geq \max_{i\in \mathcal{I}'} v_{ij'}$ for all $j' \in \mathcal{J}$.
  \label{lem:price grows}
\end{lemma}
\begin{proof}
  We have $\sum_{i \in \mathcal{I}'} u_i \leq \sum_{j' \in \mathcal{J}}f_{j'}(n) v_{j'}''$, which implies
  \begin{align*}
    \sum_{i \in \mathcal{I}'} \frac{B_i}{\beta_i} =
    \sum_{i \in \mathcal{I}'} u_i 
     \leq \sum_{j' \in \mathcal{J}}f_{j'}(n) v_{j'}''
    \Rightarrow
    \sum_{i \in \mathcal{I}'} \frac{1}{\beta_i} \leq \frac{\sum_{j' \in \mathcal{J}}f_{j'}(n) v_{j'}''}{\min_i B_i} \leq \frac{\sum_{j' \in \mathcal{J}}c_{j'}n v_{j'}''}{\min_i B_i}
  \end{align*}
  where $c_{j'}$ is an upper bound on the growth rate of the supply of item $j'$. For the above inequality to hold there must exist some $i^*$ such that
  \[
    \frac{1}{\beta_{i^*}} \leq \frac{n\sum_{j' \in \mathcal{J}}c_{j'} v_{j'}''}{|\mathcal{I}| \min_i B_i}
  \]
  which implies
  \[
    \beta_{i^*} \geq \frac{|\mathcal{I}|\min_i B_i}{n\sum_{j' \in \mathcal{J}}c_{j'} v_{j'}''}
  \]
  Now using the dual constraint $p_j \geq \beta_{i^*}v_{i^*j}$ we get
  \[
    p_j \geq \frac{|\mathcal{I}| \min_i B_i}{n \sum_{j' \in \mathcal{J}}c_{j'} v_{j'}''} v_{j'}'
  \]
\end{proof}

We now can complete the proof that any deviation gains in CEEI are of $O(\frac{1}{n}).$

\begin{proof}
  Consider some arbitrary individual $i$ with valuation vector $v_i$. 
  First we bound the gain from some specific item $j$. Pick some  $v_j'>0, v_j''$ and let $\mathcal{I}'_j$ be the set of individuals in a given market such that $v_j' \leq v_{i'j} \leq v_j''$ for all $i' \in \mathcal{I}'_j$. Let the probability of sampling an individual belonging to $\mathcal{I}'_j$ be $\theta_j$. Such a set is guaranteed to exist since $f$ has full support, and all valuations belong to some set with positive measure. When sampling a market of size $n$, the expected size of $\mathcal{I}'_j$ is $\expectation(|\mathcal{I}'_j|)=n\theta_j$. Now we can use Hoeffding's inequality to get that $|\mathcal{I}'_j|$ has size less than $\frac{1}{2}n\theta_j$ with probability at most $2e^{-\frac{1}{2}n\theta_j^2}$.
  Finally, if $\mathcal{I}'_j$ has size at least $\frac{1}{2}n\theta_j$ then by Lemmas~\ref{lem:small price effect} and~\ref{lem:price grows} 
  the possible improvement to bang-per-buck by lowering the price to $p_j$ can be bounded as (for sufficiently large $n$ such that $s_j$ is large enough to guarantee not dividing by zero)
  \begin{align}
    \frac{B_iv_{ij}}{p_j}
    \leq
    \frac{B_iv_{ij}}{p_j' - \frac{B_i}{s_j}}
    =
    \frac{B_iv_{ij}}{p_j'(1 - \frac{B_i}{s_jp_j'})}
    =
    \frac{B_iv_{ij}}{p_j'}\frac{1}{(1 - \frac{B_i}{s_jp_j'})}
    =
    \frac{B_iv_{ij}}{p_j'} \frac{s_jp_j'}{s_jp_j' - B_i}
    \leq
    \frac{B_iv_{ij}}{p_j'} \frac{s_jp_j^{\downarrow}}{s_jp_j^{\downarrow} - B_i}
    \label{eq:ceei bpb bound}
  \end{align}

  Now consider $\mathcal{I}'_j$ for each $j$. The probability that at least one $j$ is such that $\mathcal{I}'_j$ does not have at least $\frac{1}{2}n\theta_j$ individuals with such a valuation can be upper bounded by the union bound $\sum_{j\in \mathcal{J}}2e^{-\frac{1}{2}n\theta_j^2} \leq 2me^{-\frac{1}{2}n\min_j\theta_j^2}$.
  We now bound the expected utility gain by bounding each of the two cases: there exists a $j$ such that $\mathcal{I}'_j$ does not have at least $\frac{1}{2}n\theta_j$ individuals with such a valuation, or there does not exist such a $j$.

  First the case where such a $j$ exists. Since the likelihood of such a $j$ existing is exponentially decreasing, it will be sufficient to loosely bound the utility in this case by $s\cdot v_i$, the utility to $i$ of getting all the items.
  Thus the contribution to the expected gain from this case is
  \begin{align}
    2me^{-\frac{1}{2}n\min_j\theta_j^2} s\cdot v_i
    \label{eq:expected gain some j exists}
  \end{align}

  When there is no such $j$, we know from Lemma~\ref{lem:price grows} that each price is lower bounded by $p_j \geq \frac{\theta_j\min_iB_i}{2\sum_{j'}c_{j'}v_{j'}''}$. Thus the maximum utility that $i$ can gain in this case can be upper bounded by $u_i^{\uparrow} \defeq \max_j \frac{B_iv_{ij}2\sum_{j'}c_{j'}v_{j'}''}{\theta_j\min_{i'}B_{i'}}$, which is independent of the instance size $n$. Now the utility gain for any particular market where no such $j$ exists can be bounded as
  \begin{align}
     (\max_j \frac{B_iv_{ij}}{p_j} - \max_{j'}\frac{B_iv_{ij'}}{p_{j'}})
    \leq
     (\max_j \frac{B_iv_{ij}}{p_j'} \frac{s_jp_j^{\downarrow}}{s_jp_j^{\downarrow} - B_i} - \max_{j'}\frac{B_iv_{ij'}}{p_{j'}})
    \leq
    u_i^{\uparrow} (\max_j \frac{s_jp_j^{\downarrow}}{s_jp_j^{\downarrow} - B_i} - 1)
    \label{eq:expected gain no j exists}
  \end{align}

  Combining \eqref{eq:expected gain some j exists} and \eqref{eq:expected gain no j exists} the total expected gain can be bounded as
  \begin{align*}
    (1-2me^{-\frac{1}{2}n\min_j\theta_j^2}) u_i^{\uparrow} (\max_j \frac{s_jp_j^{\downarrow}}{s_jp_j^{\downarrow} - B_i} - 1)
    + 2me^{-\frac{1}{2}n\min_j\theta_j^2} s\cdot v_i
    &\leq
    O(u_i^{\uparrow} (\max_j \frac{s_jp_j^{\downarrow}}{s_jp_j^{\downarrow} - B_i} - 1)) \\
    &\leq
    O(\frac{1}{n})
  \end{align*}
\end{proof}
\subsection{Proof that EqEEI Has No Disparate Impact in Allocations}
\begin{proof}
  Consider the distributions $f_x^0, f_x^1$. For any allocation vector $x_i \in X_i$, let $\hat V_{x_i}$ be the set of valuation vectors such that agents with a valuation vector in $\hat V_{x_i}$ is assigned $x_i$. Note that for a specific valuation, only one allocation is possible, since we pool the items among all individuals with the same valuation vector and redivide their items uniformly among the individuals with that valuation vector. Now we use the fact that $\hat f_{\hat V}^1(v_i)=\hat f_{\hat V}^0(v_i)$ for all $v_i$ to get
  \[
    f_x^0(x_i) = \sum_{v_i \in \hat V_{x_i}} \hat f_{\hat V}^0(v_i)
    = \sum_{v_i \in \hat V_{x_i}} \hat f_{\hat V}^1(v_i) = f_x^1(x_i).
  \]
\end{proof}

\subsection{Proof that CEEqI is SP-L}

\begin{proof}
  Let $i$ be some individual in group $1$, and let $\delta > 0$. Let $x,p$ be the allocation and prices when $i$ reports truthfully, and let $x',p'$ be the allocation and prices when $i$ reports some scaled value $\alpha v_i$ for $\alpha \in [0,1)$. Let $B_1,B_1'$ be the corresponding budgets for group $1$. We know that $U(x)_1 = U(x)_0$, and thus when $i$ reports a valuation scaled to near zero, the CEEqI solution must increase the average utility of group $1$ by $u_i(x)$. We show that a constant fraction of the $n_1$ other individuals in group $1$ capture a linear amount of utility for any increase in utility to $i$, and thus as the market gets large the increase in utility for $i$ cannot be large since these other individuals quickly gain the requisite $u_i(x)$ utility that achieves the CEEqI condition.

  Let $\mathcal{I}_\delta$ be the set of individuals such that $|v_{ij} - v_{i'j}| \leq \delta$ for all $i' \in \mathcal{I}_\delta$. For any $i'\in \mathcal{I}_\delta$  we can lower bound the increase in utility that $i'$ gets based on the utility $\Delta u_i = u_i(x') - u_i(x)$ that $i$ gains by misreporting. Let $j,j'$ be the best bang-per-buck items for $i$ under $p$ and $p'$ respectively.
  \[
    \frac{B_1v_{i'j}}{p_j}
    \leq \frac{B_1v_{ij}}{p_j} + \frac{B_1}{p_j}\delta
    \leq \frac{B_1v_{ij'}}{p_j'} + \frac{B_1}{p_j}\delta - \Delta u_i
    \leq \frac{B_1v_{i'j'}}{p_j'} + \delta(\frac{B_1}{p_j} + \frac{B_1}{p_{j'}}) - \Delta u_i
  \]
  Thus $i'$ gains at least $\Delta u_i - \delta(\frac{B_1}{p_j} + \frac{B_1}{p_{j'}})$. 

  Now let $\theta_\delta$ be the probability of sampling an individual $i'$ such that $i' \in \mathcal{I}_\delta$. The expected size of $\mathcal{I}_\delta$ is $n\theta_\delta$. Now the probability of having less than $\frac{1}{2}n\theta_\delta$ individuals in $\mathcal{I}_\delta$ is less than $2e^{-\frac{1}{2}\theta_\delta n}$. Thus the total utility gain in $\mathcal{I}_\delta$ is at least
  \begin{align}
    \label{eq:ceeqi spl small share}
    \left(\frac{1}{2}n\theta_\delta\right) \left(\Delta u_i - \delta\left(\frac{B_1}{p_j} + \frac{B_1}{p_{j'}}\right)\right)
  \end{align}
  Now we can pick $\delta$ sufficiently small such that the second parenthesis above is strictly positive (such a $\delta$ is guaranteed to exist since we can lower-bound prices via Lemma~\ref{lem:price grows}, where the choice of $\mathcal{I}$ in Lemma~\ref{lem:price grows} can be chosen independently of $\delta$). Now by misreporting their valuation as $\alpha v_i$ causes an increase in $B_1$ such that the total utility of group $1$ increases by $u_i(x)$. But by \eqref{eq:ceeqi spl small share} $i$'s share of this increase is decreasing in $n$.

  The above shows that the utility gain from reporting a scaled-down valuation tends to zero as the market gets large. Since the theorem above shows that the utility gain from reporting a different utility vector $v_i'$ combined with some scalar $\alpha$ tends to zero, as compared to reporting $\alpha v_i$. Thus the expected utility from misreporting tends to $0$ as the market gets large.

  The argument for individuals in group $0$ is completely analogous.
\end{proof}

\subsection{MovieLens 1M Training}
We construct our market as follows. We train a model to minimize the loss $$\min_{\alpha, \beta, U, M} \sum_{dataset} (r_{ij} - \alpha_i + \beta_j + u_i \cdot m_j)^2$$ where $u_i$ is a user specific $d$ dimensional vector, $m_j$ is a movie specific $d$ dimensional vector and $\alpha_i$, $\beta_j$ are user/movie specific biases (i.e. average ratings). 

We train a recommender system model on the data using standard techniques in PyTorch using a validation set to choose hyperparameters $d$ and regularization parameter. Our best performing model has an out-of-sample MSE of $\sim .86$ (comparable to other matrix-based collaborative filtering techniques \cite{sedhain2015autorec}), dimension $d=10$ and weight decay ($l_2$ regularization) parameter of $1e-5.$ To enforce the positivity constraint we project the vectors such that each predicted $\hat{r}_{ij} > .01$. 


\end{document}